\documentclass[11pt]{article}
\usepackage[utf8]{inputenc}
\usepackage{amsfonts}
\usepackage[bottom]{footmisc}
\usepackage{xcolor}
\usepackage{qcircuit}
\usepackage{amsmath}
\usepackage{enumerate}
\usepackage{appendix}
\usepackage{hyperref}
\usepackage{amsthm}
\usepackage{multirow}
\usepackage{tikz}
\usetikzlibrary{matrix,decorations.pathreplacing}
\usepackage{mleftright}
\usepackage{amssymb}
\usepackage{algorithm}
\usepackage{algpseudocode}
\algrenewcommand\algorithmicrequire{\textbf{Input:}}
\algrenewcommand\algorithmicensure{\textbf{Output:}}
\usepackage{mathabx}
\usepackage{geometry}
\usepackage{comment}
\usepackage{subcaption}

\newcommand{\bra}[1]{\left \langle{#1} \right |}
\newcommand{\ket}[1]{\left |{#1} \right \rangle}

\newcommand{\braket}[2]{\langle{#1}|{#2}\rangle}

\newtheorem{datainput}{Data Input}

\newtheorem{lem}{Lemma}

\newtheorem{defn}{Definition} 
\newtheorem{fact}{Fact}
\newtheorem{assumption}{Assumption}

\newtheorem*{claim*}{Claim}
\newtheorem{corollary}{Corollary}

\usepackage{mathtools}

\usepackage{titlesec}
\titleformat{\paragraph}
{\normalfont\normalsize\bfseries}{\theparagraph}{1em}{}
\titlespacing*{\paragraph}
{0pt}{3.25ex plus 1ex minus .2ex}{1.5ex plus .2ex}

\definecolor{Pr}{rgb}{0.4,0.3,0.9}

\DeclareMathOperator*{\argmax}{arg\,max}

\newtheorem{theorem}{Theorem}
\usepackage{geometry}
\def\be{\begin{eqnarray}}
\def\ee{\end{eqnarray}}

\title{}
\author{}
\date{}

\begin{document}

\title{A Quantum Online Portfolio Optimization Algorithm}

\author{Debbie Lim\footnote{Centre for Quantum Technologies, {\tt limhueychih@gmail.com}}, Patrick Rebentrost\footnote{Centre for Quantum Technologies, {\tt cqtfpr@nus.edu.sg}}}

\date{}

\maketitle

\begin{abstract}
Portfolio optimization plays a central role in finance to obtain optimal portfolio allocations that aim to achieve certain investment goals. Over the years, many works have investigated different variants of portfolio optimization. Portfolio optimization also provides a rich area to study the application of quantum computers to obtain advantages over classical computers. 
In this work, we give a sampling version of an existing classical online portfolio optimization algorithm by Helmbold \textit{et al.}, for which we in turn develop a quantum version. The quantum advantage is achieved by using techniques such as quantum state preparation, inner product estimation and multi-sampling. Our quantum algorithm provides a quadratic speedup in the time complexity, in terms of $n$, where $n$ is the number of assets in the portfolio. The transaction cost of both of our classical and quantum algorithms is independent of $n$ which is especially useful for practical applications with a large number of assets. 
\end{abstract}

\section{Introduction}
\subsection{Online optimization}\label{convex_online_optimization_intro}

Online optimization is a branch of optimization, where the input data is revealed over time and decisions have to made while having incomplete knowledge about the input data. At every time step, a loss function will be given based on the decisions made so far. A feature of online optimization is that the sequential input can be given in an adversarial manner; the provable guarantees hold even if the input is chosen by an adversary who knows the algorithm's strategy. Online convex optimization studies the problem of optimizing a convex function over a convex set in an online fashion. The popular first-order algorithms for online convex optimization include variants of gradient descent, mirror descent and coordinate descent \cite{Zinkevich2003a, Hazan2016b, Flaxman2004, Hazan2007, Wang2014}. 

Apart from the commonly known gradient descent method, the multiplicative weight update method is another alternative to solving optimization problems. The multiplicative weight update method is a primal-dual algorithm proposed by Arora and Kale \cite{Arora2016a}, which assigns an initial weight to each expert and at every iteration, updates the weights according to the experts' performances. This algorithm can also be extended to the online convex optimization framework when the convex set is the $n$-dimensional simplex.  The multiplicative weight update method is one of the second-order methods in online convex optimization besides the Newton's method \cite{Schraudolph2007}, which iteratively finds the roots of a differentiable function. Some applications of the 
multiplicative weight update method include solving linear programs and semidefinite programs \cite{Arora2016}, learning algorithms \cite{Helmbold2009}, and portfolio selection \cite{Helmbold1998}.

In zeroth-order online convex optimization (bandit convex optimization), the feedback is in the form of a real number (instead of a loss function), thereby being less informative. The first algorithm for bandit convex optimization was proposed by Flaxman \textit{et al.} \cite{FKM}. Subsequently, many follow up works \cite{pmlr-v49-bubeck16, 10.1145/3453721, Lattimore2020} have been done to improve the regret bound. 

\subsection{Portfolio optimization}
Portfolio optimization is a standard problem in mathematical finance. The first formalization, the Markowitz (mean-variance) model, is proposed by Nobel prize winner, Harry  Markowitz \cite{HarryMarkowitz1952}. It is a single-period unconstrained quadratic programming problem, which  either maximizes the portfolio return for a given level of risk or minimizes the risk for a given return. However, there are several caveats concerning the implementation of this model. Among them, the model relies on the knowledge of the mean and covariance matrix of the asset returns.
Besides that, the model suffers from error maximization, i.e., a small change in the inputs can result in a large change in the portfolio \cite{ElGhaoui:M00/59}. Consequently, many refinements have been proposed to make the model more realistic \cite{Skolpadungket2007, Ponsich2012, Yen2014, Kalayc2017, Khan2021, Unknowna, King1993, Konno1991, Mills1997, Turnbull1990ARO, Morgan1996, Rockafellar2002}.

Reference~\cite{Helmbold1998} by Helmbold \textit{et al.} is a seminal paper discussing a (classical) online algorithm for portfolio selection based on the multiplicative weight update rule. The update rule was derived using a framework introduced by Reference~\cite{Kivinen1995} for online regression. The authors adapted this framework to the online portfolio selection setting  and the resulting algorithm uses linear (in the number of assets) time and space to update the portfolio vector at each time step. 
A survey on (classical) online portfolio selection was done by Reference~\cite{Li2014} from an online machine learning perspective. The survey paper expressed online portfolio selection as a sequential decision problem and included various classes of related algorithms, such as follow the winner, follow the loser, pattern-matching-base approaches and meta-learning algorithms. 

\subsection{Our work}
Our main contribution is an online quantum algorithm for portfolio selection. We show that the online portfolio selection algorithm  proposed by Helmbold \textit{et al.} \cite{Helmbold1998} can be quantized. We adopt a step by step approach to demonstrate how we arrive at the quantum algorithm. We start from Algorithm~\ref{algo_helmbold}, the slightly extended version of the classical online portfolio optimization algorithm from Reference~\cite{Helmbold1998} which includes a transaction cost (Corollary~\ref{Helmbund_regretbound}). Next, we implement a sampling procedure in Algorithm~\ref{algo_helmbold_sampling} which renders the transaction cost independent of $n$ (Theorem~\ref{regret-bound_helmold_sampling}). Subsequently, we build on Algorithm~\ref{algo_helmbold_sampling} but use an inner product estimation procedure to compute the portfolio vectors in Algorithm~\ref{algo_helmbold_sampling_inner} (Corollary~\ref{apprx_Helmbund_regretbound}). Lastly, we 
use quantum inner product estimation and quantum multi-sampling to replace their classical counterparts and use quantum state preparation to prepare the portfolio vector when devising our quantum online portfolio optimization algorithm, Algorithm~\ref{algo} (Theorem~\ref{thmQuantum}).

We summarize our results in the table below: 

\begin{table}[h]
\centering 
\begin{tabular}{ |p{3cm}|c|c|c|c| }
\hline Name & Alg. & Regret & Run time & Transaction cost\\
\hline
Online  & \ref{algo_helmbold} &  $\frac{1}{r_{\min}}\sqrt{\frac{\log n}{2T}}$ & $O(Tn)$ & $O(TnC)$\\
\hline Sampling-based Online & \ref{algo_helmbold_sampling} & $\frac{2}{r_{\min}}\sqrt{\frac{\log n}{2T}}$ & $O(T^2n\log \frac{T}{\delta})$ & $O(T^2 C\log \frac{T}{\delta})$\\
\hline 
Approximate Sampling-based Online & \ref{algo_helmbold_sampling_inner} & $\frac{8}{r_{\min}}\sqrt{\frac{\log n}{2T}}$ & $O\left(T n  + \frac{T^2}{r_{\min}}\log\frac{T}{\delta}\right)$  & $O(T^2 C\log\frac{T}{\delta})$\\
\hline 
Quantum Online & \ref{algo} & $\frac{12}{r_{\min}}\sqrt{\frac{\log n}{2T}}$ & $\tilde O\left(T^3\sqrt{\frac{n}{r_{\min}}}\log^{1.5}\left(\frac{1}{\delta}\right)\right)$ & $O(T^2 C\log\frac{T}{\delta})$\\
\hline 
\end{tabular}\caption{Summary of results. Throughout this work, $n$ is the number of assets, $T$ is the total number of time steps, and $r_{\min}$ is the lower bound for the price relatives (see Assumption \ref{assumption_r_min}). In addition, $C$ is the transaction cost (see Assumption \ref{assumption_cost}) and $3\delta$ is an upper bound on the probability of failure.}

\end{table}

The regret bound achieved by Algorithm~\ref{algo} is larger than than that of Algorithm~\ref{algo_helmbold} only by a small factor, and the algorithm provides a quadratic speedup in the run time in terms of $n$, the number of assets in the portfolio. The speedup is due to the use of amplitude amplification, quantum inner product estimation, and quantum multi-sampling. In addition, the algorithm does not have to store the portfolio vectors $w^{(t)}$ explicitly for every time step $t$. Instead, the portfolio vectors can be computed efficiently via  unitaries that perform arithmetic operations. Moreover, the transaction cost of our algorithm is independent of $n$, which is especially useful for practical applications with a large number of assets in the portfolio. 
 
\subsection{Related work}\label{Related_Work}
References \cite{Orus2019,  Bouland2020, Egger2021, Herman2022} discuss the state-of-the-art, potential, and challenges of quantum computing in finance. 
Rosenberg \textit{et al.} \cite{rosenberg2016solving} discuss a non-convex discrete portfolio optimization problem, in the context of D-Wave's quantum annealer. Their problem formulation aims to maximize the expected return while minimizing the risk and transaction costs. They numerically showed that the quantum annealer in principle could solve this problem with high probability and this success probability can be increased by making adjustments to the annealer.  
In Ref.~\cite{Rebentrost2018a}, the authors proposed a quantum algorithm for the unconstrained  portfolio optimization problem. The algorithm uses quantum linear system solvers \cite{Harrow_2009,childs2017quantum} to obtain speedups for portfolio optimization problems that can be reduced to unconstrained quadratic programs, which in turn are reducible to a single linear system. Subsequently, Ref.~\cite{Kerenidis2019} gave a quantum algorithm for the general constrained portfolio optimization problem with an arbitrary number of nonnegativity and budget constraints, resulting in a polynomial speedup in terms of the number of assets, as compared to the best known classical algorithm when only a moderately accurate solution is required.    

In terms of practical implementation, Reference~\cite{Hodson2019} evaluated the experimental performance of using the Quantum Approximate Optimization Algorithm and the Quantum Alternating Operator Ansatz to solve a discrete portfolio optimization problem for a multi-period portfolio rebalancing setting. Subsequently, Ref.~\cite{Slate20} numerically showed that the Quantum Walk Optimization Algorithm is capable of achieving a significantly better performance. 
In the noisy intermediate-scale quantum (NISQ) setting, the work by Ref.~\cite{YalovetzkyNISQ-HHLHardware} proposed a hybrid algorithm for end-to-end execution of small scale portfolio optimization problems on near-term devices. Their algorithm uses techniques such as mid-circuit measurement, quantum conditional logic, and qubit reset/reuse, and also improved on the existing eigenvalue inversion component of HHL. 

Online optimization has been considered in the quantum setting. Boosting is an approach to improve the performance of a weak learning algorithm in terms of its accuracy. Quantum boosting was discussed in Reference~\cite{arunachalam2020quantum} to improve the time complexity of the widely used classical AdaBoost proposed by Ref.~\cite{freund1999short}. Subsequently, a follow-up work by Ref.~\cite{izdebski2020improved} was done to provide a significantly faster and simpler quantum boosting algorithm. 
The Hedge algorithm proposed by Freund and Schapire uses the multiplicative weight update method to adaptively allocate mixed strategies to solve an adversarial online optimization problem. The Sparsitron by Reference~\cite{klivans2017learning} which is based on the Hedge algorithm, is a machine learning algorithm for  undirected graphical models. Quantum versions of both the Hedge algorithm and the Sparsitron were discussed in Ref.~\cite{Rebentrost2021}. 
In zeroth-order optimization, there are instances where quantum advantage have been proven. For example, in the multi-armed bandits setting, Ref.\cite{wang2021quantum} proposed a quantum algorithm that provides an exponential speedup in terms of the time $T$ in the regret bound as compared to the well known classical lower bounds \cite{auer2002nonstochastic, lattimore2020bandit}. In bandit convex optimization, Ref.~\cite{He2022} gave an quantum algorithm that achieves a regret bound that is independent of $n$, the dimension. This  outperforms the best known optimal classical algorithm \cite{Shamir2017}. 

\section{Preliminaries}\label{Preliminaries}
\subsection{Notations}
We use $[n]$ to represent the set $\{1, \cdots, n\}$, where $n\in\mathbb{Z}_+$ and denote the $i$-th entry of a vector $v\in\mathbb{R}^n$ as $v_i$ for $i \in [n]$. If a vector has a time dependency we denote it as $v^{(t)}$. Let $\textbf{e}_i$ be the vector of all zeros with a 1 in the $i$-th position. The $\ell_1$-norm of a vector $v\in\mathbb{R}^n$ is defined as $\lVert v\rVert_1 \coloneqq \sum_{i=1}^n |v_i|$. We use $\bar{0}$ to denote the all zeros vector and use $\ket{\bar{0}}$ to denote the state $\ket{0}\otimes\cdots \otimes\ket{0}$, where the number of qubits is clear from the context. The maximum entry in absolute value of a vector $v\in\mathbb{R}^n$ is denoted as $\lVert v\rVert_{\max}=\displaystyle\max_{i\in[n]} |v_i|$ and we denote the maximum entry of a vector $v\in\mathbb{R}^n$ as $v_{\max} = \displaystyle\max_{i\in[n]} v_i$. For $v\in\mathbb{R}^n$  and $k\in\mathbb{R}$, $k^v\in\mathbb{R}^n$ is the element-wise exponential $v$, i.e. $(k^v)_i = k^{v_j}$.  We write the natural logarithm (base $e$) as $\log$.  We use $\tilde{O}(\cdot)$ to hide the polylog factor, i.e., $\tilde{O}(f(n))= O(f(n)\cdot {\rm polylog}(f(n)))$. We sometimes use $O(1)$ to denote  a constant.

\subsection{The computational model}\label{input_access}

We refer to the run time of a classical/quantum computation as the number of basic gates performed. We assume a classical arithmetic model, which allows us to ignore issues arising from the fixed-point representation of real numbers. The basic arithmetic operations take constant time.
In the quantum setting, we assume a quantum circuit model. Each quantum gate in the circuit represents an elementary operation, and the application of each quantum gate takes constant time. The time complexity of a given unitary operator $U$ is the minimum number of basic quantum gates required to prepare $U$.
In addition, we assume a quantum arithmetic model, which is equivalent to the classical model in that arithmetic operations take constant time.   
Our quantum algorithm assumes quantum query access to certain vectors. 
For the oracles, the representation of real numbers to finite precision is also not taken into account. 
Given a vector $v \in \mathbb{R}^n$, we say we have quantum query access to this vector if we have access to the operation $O_v$ which performs
\be
O_v\ket{i}\ket{\bar 0} & = & \ket{i}\ket{v_i}.
\ee
The second register is assumed to contain sufficient qubits to make all the subsequent computations accurate, in analogy to the sufficient bits that a classical algorithm assumes to run correctly.

\section{The online portfolio optimization framework}\label{framework}
Consider $T$ discrete time  steps and $n$ assets, and the setting as in Ref.~\cite{Helmbold1998}. 
A portfolio of these $n$ assets at time $t\in[T]$ is described by a vector $w^{(t)}$ such that for each $i\in [n]$, $w_i^{(t)}\geq 0$ and $\displaystyle\sum_{i=1}^n w_i^{(t)}=1$. Here, we make the no-shortselling assumption, see also below. 
Each asset has a price as a function of time and in this work we consider the time series of closing prices. The original paper \cite{Helmbold1998} uses the opening prices, and we assume that the closing price at $t$ is the same as the opening price at $t+1$.
Define the day-to-day return as
\begin{equation}
R_i^{(t)} := \frac{\text{closing price of asset } i \text{ on day } t}{\text{closing price of asset } i \text { on day } t-1}.
\end{equation}
In this work, the performance of the assets is reflected in a price relative vector $\rho^{(t)}\in\mathbb{R}^n_+$, where for all $i\in[n]$, $\rho_i^{(t)}$ is the ratio  
\begin{equation}
\rho_i^{(t)} := \frac{R_i^{(t)}}{R_{\max}^{(t)}},
\end{equation}
where $R_{\max}^{(t)} = \displaystyle\max_{j\in[n]} R_j^{(t)}$.
By definition, $0\leq \rho_i^{(t)}\leq 1$ for all $i\in[n]$ and $t\in[T]$. However, we assume a known lower bound $r_{\min} \in (0, 1]$ such that $0 < r_{\min} \leq \rho_i^{(t)}$ for all $i\in[n]$ and $t\in[T]$. Given $w$ and $\rho$, an investor's wealth changes by a factor of 
\be
w\cdot \rho & = & \displaystyle\sum_{i=1}^n w_i\cdot \rho_i
\ee
from one trading day to the next. In the online portfolio selection setting, the learning algorithm has access to the price relative vectors 
$\rho^{(1)}, \cdots, \rho^{(t)}$ at the end of trading day $t$. The algorithm then selects the portfolio $w^{(t+1)}$ for the next day. At the end of each trading day $t$, $\rho^{(t)}$ is revealed and the investor's wealth changes by a factor of $w^{(t)}\cdot \rho^{(t)}$. As time progresses, $\rho^{(1)}, \cdots, \rho^{(T)}$ will be revealed and $w^{(1)}, \cdots, w^{(T)}$ will be selected. From the start of trading day $1$ through the start of trading day $(T+1)$, the wealth changes by a factor of 
\be
S & = & \displaystyle\prod_{t=1}^T w^{(t)}\cdot \rho^{(t)}.
\ee
Similar to the analysis in Reference~\cite{Helmbold1998}, we will deal with the normalized logarithm of $S$: 
\be
LS & = & \frac{1}{T}\displaystyle\sum_{t=1}^T \log\left (w^{(t)}\cdot \rho^{(t)}\right),
\ee
since wealth often grows or decays geometrically in typical markets. 

Consider the ``offline gain", $LS^*=\frac{1}{T}\displaystyle\max_{\{w:\lVert w\rVert_1=1\}} \sum_{t=1}^T \log\left(w\cdot \rho^{(t)}\right)$, which is the maximum gain in wealth achievable when choosing the same portfolio $w^*=\displaystyle\argmax_{\{w:\lVert w\rVert_1=1\}} S(w)=\displaystyle\argmax_{\{w:\lVert w\rVert_1=1\}} LS(w)$ for all trading days $t\in [T]$. 
The difference of offline loss and the loss of some sequence of normalized $w^{(t)}$ is called regret. Formally, it is $LS^* - LS$ and can be naively bounded as $LS^* - LS\leq \log\left(\frac{1}{r_{\min}}\right)$. The bound follows from
\be
LS^* - LS & = & \max_{\{w:\lVert w\rVert_1=1\}} \frac{1}{T}\displaystyle\sum_{t=1}^T \log(w\cdot\rho^{(t)}) - \frac{1}{T}\displaystyle\sum_{t=1}^T \log(w^{(t)}\cdot\rho^{(t)}) \\
& \leq & \frac{1}{T}\displaystyle\sum_{t=1}^T \log\left(\lVert \rho^{(t)}\rVert_{\max} \right) - \frac{1}{T}\displaystyle\sum_{t=1}^T \log\left(\lVert w^{(t)}\rVert_1 r_{\min}\right)\\
& = & \log\left(\frac{1}{r_{\min}}\right).
\ee
This bound does not decrease with $T$. A main result of the work by  Reference~\cite{Helmbold1998} is a sequence of $w^{(t)}$ which shows a regret bound of about $1/\sqrt T$.
 
We would like to emphasize again the following assumptions.
\begin{assumption}
We assume that there is no short-selling throughout the trading period. Therefore,  $w^{(t)}_i \geq 0$ for all $t\in[T]$ and $i\in[n]$. 
\end{assumption}
To model the cost of trading, we assume a fixed transaction cost per investment. This cost will highlight the difference between the standard and the sampling algorithm. It was mentioned as a possible extension in Reference~\cite{Helmbold1998}.
\begin{assumption}\label{assumption_cost}
We assume that a transaction cost of $C\geq 0$ is incurred when investing in a single asset. Here, this transaction cost is independent of the amount of asset that is bought. 
\end{assumption}
The following assumption is as in the classical work and simplifies the analysis of the regret bound. Reference~\cite{Helmbold1998} also relaxes this assumption and provides a different regret bound for the relaxed setting.  
\begin{assumption}\label{assumption_r_min} We assume a known lower bound $r_{\min} \in (0,1]$ for the price relatives, i.e., 
$r_{\min}\leq \rho_i^{(t)}\leq 1$ for all $i\in[n]$ and $t\in[T]$.
\end{assumption}

\subsection{Helmbold \textit{et al.}'s algorithm}
In Reference~\cite{Helmbold1998}, the authors provide a online algorithm for portfolio optimization. Given a current portfolio $w^{(t)}$, consider the following optimization problem 
\begin{equation}\label{opt}
\max_{\{w^{(t+1)}:\lVert w^{(t+1)}\rVert_1=1\}} \eta\left(\log \left(w^{(t)}\cdot \rho^{(t)}\right)+\frac{\rho^{(t)}\cdot \left(w^{(t+1)} - w^{(t)}\right)}{w^{(t)}\cdot \rho^{(t)}}\right) - d\left(w^{(t+1)}, w^{(t)}\right).
\end{equation}
The problem is to pick a portfolio vector $w^{(t+1)}$ that maximizes the gain and at the same time, is close to the portfolio vector picked in the previous iteration. Here, $\eta$ is the ``learning rate" and $d\left(w^{(t+1)}, w^{(t)}\right)$ is a distance measure between $w^{(t+1)}$ and $w^{(t)}$. We formally define the update rule which is the solution to Eq.~(\ref{opt}) when the relative entropy is used as the distance measure. 

\begin{defn}[Exponentiated gradient $EG(\eta)$ update \cite{Helmbold1998}]\label{update_rule}
With $\eta>0$, $w \in \mathbb R^n$ and $\rho \in \mathbb R^n$, we define by $EG(\eta, w, \rho) \in \mathbb R^n$ the mapping which performs the following weight update for all $i\in[n]$:
\begin{equation}
(\eta, w, \rho, i) \to  \frac{w_i\exp\left(\eta\frac{\rho_i}{w\cdot \rho}\right)}{Z},
\end{equation}
where $Z = \displaystyle\sum_{j=1}^n w_j\exp\left(\eta\frac{\rho_j}{w\cdot \rho}\right)$. We use the shorthand notation $EG(\eta)$ if the other inputs are clear from the context.
\end{defn}

In order to solve the portfolio selection problem, Reference~\cite{Helmbold1998} gave Algorithm~\ref{algo_helmbold}, which uses linear time and space (in $n$) to update $w^{(t)}$ for each $t$. We present a slightly extended version of their algorithm by including the transaction cost for investing in an asset. 

\begin{algorithm}
\caption{Online Portfolio Optimization Algorithm}
\label{algo_helmbold}
\begin{algorithmic}[1]
\Require $n$, $\eta$, $T$, 
$C$. 
\State Initialize $w^{(1)}=\left (\frac{1}{n}, \cdots, \frac{1}{n}\right)\in\mathbb{R}^n$.
\For {$t=1$ to $T$}
\State Invest in all assets according to $w^{(t)}$, with cost $C$ for each asset. 
\State Wait until end of day. 
\State Receive price relative vector $\rho^{(t)}$. 
\For {$i=1$ to $n$}
\State $w^{(t+1)}_i\gets EG(\eta, w^{(t)},  \rho^{(t)})_i$.
\EndFor 
\EndFor
\Ensure $LS_{\rm EG} := \frac{1}{T}\displaystyle\sum_{t=1}^T \log{\left(w^{(t)}\cdot\rho^{(t)}\right)} $. 
\end{algorithmic}
\end{algorithm}

The following theorem by Reference~\cite{Helmbold1998} bounds the difference in wealth gained when using a fixed portfolio vector versus the update rule Def.~(\ref{update_rule}) applied to portfolio vector initialized to be the uniform vector $\left(\frac{1}{n}, \cdots, \frac{1}{n}\right)\in\mathbb{R}^n$.

\begin{theorem}[\cite{Helmbold1998}]\label{Helmbound_bound}
Let $u\in\mathbb{R}^n_+$ be a portfolio vector, and let $\rho^{(1)}, \cdots, \rho^{(t)}$ be of price relatives with $\displaystyle\max_{i\in[n]} \rho^{(t)}_i=1$ and $\rho^{(t)}_i\geq r_{\min}>0$ for all $i\in[n], t\in[T]$, where it is assumed that $r_{\min}$ is known. Set $w^{(1)}=\left(\frac{1}{n}, \cdots, \frac{1}{n}\right)$ and $\eta = 2r_{\min}\sqrt{\frac{2\log n}{T}}$. The update in Def.~(\ref{update_rule}) produces portfolio vectors that achieve the following bound: 
\be
\displaystyle\sum_{t=1}^T \log\left(w^{(t)}\cdot \rho^{(t)}\right) & \geq & \sum_{t=1}^T \log\left(u\cdot \rho^{(t)}\right) - \frac{\sqrt{2T\log n}}{2r_{\min}}.
\ee
\end{theorem}

Thm.~\ref{Helmbound_bound} implies the corollary below. 
\begin{corollary}[Guarantee and run time of Algorithm~\ref{algo_helmbold} \cite{Helmbold1998}]\label{Helmbund_regretbound}
Algorithm~\ref{algo_helmbold} with $\eta = 2r_{\min}\sqrt{\frac{2\log n}{T}}$ achieves
\be
LS^* - LS_{\rm EG} & \leq & \frac{1}{r_{\min}} \sqrt{\frac{\log n}{2T}}.
\ee
with a total run time of $O(Tn)$ and a transaction cost of $O(TnC)$.
\end{corollary}
Compared to the naive bound $LS^* - LS_{\rm EG}\leq \log\left(\frac{1}{r_{\min}}\right)$, this bound decreases with $T$ and is better when $T\geq \frac{\log n}{2r_{\min}^2 \log^2\left(\frac{1}{r_{\min}}\right)}$. 
\subsection{Sampling-based 
online portfolio optimization algorithm}
We now consider including a sampling procedure, which leads to a reduction in the total transaction cost  as we only invest in the sampled assets.   
\begin{fact}[$\ell_1$-sampling \cite{Vose1991, iet:/content/journals/10.1049/el_19740097}]\label{sampling_data_structure}
Given a probability vector $p\in[0, 1]^n$, there exists a data structure that samples the index $i\in[n]$ with probability $p_i$ which can be constructed in $O(n)$ time. The time required for obtaining one sample is $O(1)$. 
\end{fact}
The assumptions of Ref.~\cite{Vose1991} allow us to omit $\log (n)$ factors in the time for construction and sampling, and we adopt the same assumption in this work. 
Based on this data structure, we construct Algorithm~\ref{algo_helmbold_sampling}. This algorithm only samples multiple assets from the portfolio vector and invests only in those assets. For the portfolio update, however, the complete vector of price relatives is used and the complete new portfolio vector is computed.  
\begin{algorithm}
\caption{Sampling-based Online Portfolio Optimization Algorithm}
\label{algo_helmbold_sampling}
\begin{algorithmic}[1]
\Require $n$, $s$, $\eta$, $T$, $C$.
\State Initialize ${w}^{(1)}=\left(\frac{1}{n}, \cdots, \frac{1}{n}\right)\in\mathbb{R}^n$.
\For {$t=1$ to $T$}
\State Prepare sampling data structure for ${w}^{(t)}$ using Fact \ref{sampling_data_structure}. 
\State Sample $i^{(t)}_1, \cdots, i^{(t)}_s$ from ${w}^{(t)}$.
\State Invest the amount $1/s$ in each asset $i^{(t)}_1, \cdots, i^{(t)}_s$ at cost $C$ each. 
\State Wait until end of day. 
\State Receive price relative vector $\rho^{(t)}$.
\For {$i=1$ to $n$}
\State $w^{(t+1)}_i\gets EG(\eta, w^{(t)},  \rho^{(t)})_i$.
\EndFor 
\EndFor
\Ensure $LS_{\text{samp}}\coloneqq  \frac{1}{T}\displaystyle\sum_{t=1}^T \log\left(\frac{1}{s}\sum_{\ell=1}^s \rho^{(t)}_{i^{(t)}_\ell}\right)$
\end{algorithmic}
\end{algorithm}
The following theorem gives a upper bound on the regret of the logarithmic wealth obtained from sampling from the exponential gradient update.
\begin{theorem}\label{regret-bound_helmold_sampling}
Let $\delta\in(0,1/2)$ and $LS_{\rm EG}$ as in Algorithm~\ref{algo_helmbold}. Algorithm~\ref{algo_helmbold_sampling} outputs $LS_{\text{samp}}$ with $LS_{EG} - LS_{\text{samp}} \leq \frac{1}{r_{\min}}\sqrt{\frac{1}{2T}}$ with success probability at least $1-2\delta$. With $\eta = 2r_{\min}\sqrt{\frac{2\log n}{T}}$, the regret bound is 
\be 
LS^* - LS_{\rm samp} \leq  \frac{2}{r_{\min}}\sqrt{\frac{\log n}{2 T}},
\ee
with success probability at least $1-2\delta$. The total run time is $O\left(T^2n\log\frac{T}{\delta}\right)$ and the transaction cost is $O\left(T^2C\log\frac{T}{\delta}\right)$.
\end{theorem}
\begin{proof}
Let $s \in \mathbb Z_+$. Sample $i_\ell \in[n]$ with probability $w_{i_\ell}^{(t)}$ for all $l \in[s]$.
Define the random variable $Z^{(t)} = \frac{1}{s}\displaystyle\sum_{\ell=1}^s \rho^{(t)}_{i_\ell}$. Then, its expectation is, using the shorthand notation $\mathbb E = \mathbb E_{i_1,\cdots,i_s}$,
\be
\mathbb E\left[Z^{(t)}\right] = \mathbb E\left[\frac{1}{s}\sum_{\ell=1}^s \rho^{(t)}_{i_\ell}\right] = \frac{1}{s}\sum_{\ell=1}^s \mathbb E\left[\rho^{(t)}_{i_\ell}\right] = \mathbb E_i\left[\rho^{(t)}_{i}\right] = \sum_{i=1}^n w^{(t)}_{i}\cdot \rho^{(t)}_{i} = w^{(t)}\cdot\rho^{(t)}.
\ee
Using Hoeffding's inequality (see Fact~\ref{Hoeffding}) with $q>0$, we obtain
\be\label{bound}
\mathbb P\left[\left\vert s\cdot Z^{(t)} - s\cdot \mathbb E\left[Z^{(t)}\right]\right\vert\geq sq\right]\leq 2e^{-\frac{sq^2}{(1-r_{\min})^2}}\leq \frac{2\delta}{T},
\ee
when we set $q = \sqrt{\frac{1}{2T}}$ and $s = 2T(1-r_{\min})^2\log\frac{T}{\delta}$.
For the success probability, we  hence obtain $1-\frac{2\delta}{T}$. Now we bound
\be
LS_{EG} - LS_{\text{samp}} 
& = & \frac{1}{T}\displaystyle\sum_{t=1}^T \log\left(w^{(t)}\cdot\rho^{(t)}\right) - \frac{1}{T}\displaystyle\sum_{t=1}^T \log\left(\frac{1}{s}\displaystyle\sum_{\ell=1}^s \rho^{(t)}_{i_\ell}\right)\\
& \leq & \frac{1}{T}\displaystyle\sum_{t=1}^T \left\vert \log\left(w^{(t)}\cdot\rho^{(t)}\right) - \log\left(\frac{1}{s}\displaystyle\sum_{\ell=1}^s \rho^{(t)}_{i_\ell}\right)\right\vert\\
\text{(by Lipschitz continuity)} & \leq & \frac{1}{Tr_{\min}}\displaystyle\sum_{t=1}^T \left\vert w^{(t)}\cdot\rho^{(t)} - \frac{1}{s}\displaystyle\sum_{\ell=1}^s \rho^{(t)}_{i_\ell}\right\vert\\
\text{(by Eq.(\ref{bound}))} & \leq & \frac{1}{r_{\min}}\sqrt{\frac{1}{2T}},
\ee
with probability $1 - 2\delta$ by the union bound. Therefore, the regret is bounded by 
\be
LS^* - LS_{\text{samp}} & = &  LS^* - LS_{\rm EG} + LS_{\rm EG} - LS_{\text{samp}} \\
& \leq & \frac{1}{r_{\min}}\sqrt{\frac{\log n}{2T}} + \frac{1}{r_{\min}}\sqrt{\frac{1}{2T}}\\
& \leq & \frac{2}{r_{\min}}\sqrt{\frac{\log n}{2T}},
\ee
which holds with probability at least $1-2\delta$.
\end{proof}

The performance of the algorithm worsens slightly in three regards.
Firstly, we obtain a constant factor to the regret bound. 
Secondly, the algorithm is probabilistic and we obtain a $\log 
\left(\frac{1}{\delta}\right)$ dependence in the run time, where $2\delta$ is the failure probability of the algorithm. Usually, this failure probability can be taken as some small constant such as $0.001$ to obtain a $99.9\%$ confidence that the algorithm ran correctly. 
Thirdly, we obtain a $T^2$ dependence in the run time due to the multi-sampling step. 
The benefit of the algorithm is that the transaction cost is reduced from $Tn C$ to $O\left(T^2C\log\left(\frac{T}{\delta}\right)\right)$. 

\subsection{Convergence theorem for erroneous updates}
Before we move to an approximate classical algorithm and our quantum algorithm, we generalize the convergence result from the original work. The generalizations are in terms of the availability of the inner product and the normalization factor, both of which will be known only approximately in the quantum algorithm.
The generalization is embodied in the following definition of an erroneous update rule.
\begin{defn}[Erroneous exponentiated gradient update]\label{update_rule_erroneous_norm}
Let $\eta>0$, $w \in \mathbb R^n$, and $\rho \in \mathbb R^n$. In addition, let $\epsilon_I \in(0,1/2)$ and $\tilde I >0$ such that $\vert \tilde I - I \vert \leq I \epsilon_I$, where $I:= w\cdot \rho$.
Moreover, let $\epsilon_Z \in (0,1/2)$ and $\tilde Z >0$ such that $\vert \tilde Z - Z \vert \leq Z \epsilon_Z$,
where $Z = \displaystyle\sum_{j=1}^n w_j\exp\left(\eta\frac{\rho_j}{\tilde I}\right)$.
We define the erroneous weight update $EEG(\eta, w, \rho, \tilde I, \tilde Z) \in \mathbb R^n$ as the mapping which computes for $i \in [n]$
\begin{equation}
(\eta, w, \rho, \tilde I, \tilde Z,i) \to \frac{w_i\exp\left(\eta\frac{\rho_i}{\tilde I}\right)}{\tilde Z}.
\end{equation}
We use the shorthand notation $EEG(\eta)$ if the other inputs are clear from the context.
\end{defn}
The main theorem for this update rule is as follows, for which we modify the proof of Theorem \ref{Helmbound_bound} from Ref.~\cite{Helmbold1998}.
\begin{theorem}[Main convergence theorem for erroneous updates]\label{EG_Helmbound_bound}
Let $u\in\mathbb{R}^n_+$ be a portfolio vector, and let $\rho^{(1)}, \cdots, \rho^{(t)}$ be price relatives with $\max_{i\in[n]} \rho^{(t)}_i=1$ and $\rho^{(t)}_i\geq r_{\min}>0$ for all $i\in[n], t\in[T]$.
With $w^{(1)}=\left(\frac{1}{n}, \cdots, \frac{1}{n}\right)$ and $\eta = 2r_{\min}\sqrt{\frac{2\log n}{T}}$, the update in Def.~(\ref{update_rule_erroneous_norm}) produces portfolio vectors that achieve the bound
\be
\displaystyle\sum_{t=1}^T \log\left(w^{(t)}\cdot \rho^{(t)}\right) &\geq& \sum_{t=1}^T \log\left(u\cdot \rho^{(t)}\right) - \frac{3\sqrt{2T\log n}}{r_{\min}},
\ee
when $\epsilon_I = \frac{3\eta}{4 r_{\min}}$ and $\epsilon_Z = 0$, 
and
\be 
\displaystyle\sum_{t=1}^T \log\left(w^{(t)}\cdot \rho^{(t)}\right)&\geq& \sum_{t=1}^T \log\left(u\cdot \rho^{(t)}\right) - \frac{5\sqrt{2T\log n}}{r_{\min}},
\ee
when $\epsilon_I = \frac{3\eta}{4 r_{\min}}$ and $\epsilon_Z = \frac{\eta^2}{r^2_{\min}}$.
\end{theorem}
\begin{proof}
Let $D\left(u\lVert v\right) \coloneqq \displaystyle\sum_{i=1}^n u_i\log \left(\frac{u_i}{v_i}\right)$. Fix $w^{(t)}$ and let $w^{(t+1)} = EEG(\eta, w^{(t)},\rho^{(t)},\tilde I^{(t)}, \tilde Z^{(t)})$, with $\tilde I^{(t)},\tilde Z^{(t)}$ as in Def.~(\ref{update_rule_erroneous_norm}). 
Let $\Delta_t = D(u\lVert w^{(t+1)}) - D(u\lVert w^{(t)})$. Then 
\be\label{change}
\Delta_t & = & -\displaystyle\sum_{i=1}^n u_i\log\left (\frac{w_{i}^{(t+1)}}{w_{i}^{(t)}}\right)\\
& =  & -\displaystyle\sum_{i=1}^n u_i \left(\eta\left(\frac{\rho^{(t)}_i}{\tilde I^{(t)} }\right) - \log \tilde Z^{(t)} \right)\\
& = & -\eta\frac{u\cdot \rho^{(t)}}{\tilde I^{(t)}} + \log \tilde Z^{(t)}\\
& \leq & -\eta\left(1-\frac{4}{3}\epsilon_I\right)\frac{u\cdot \rho^{(t)}}{I^{(t)}} + \log \tilde Z^{(t)}. 
\ee
Since $\rho^{(t)}_i\in[0,1]$ and  $k^y\leq 1-1(1-k)y$ for $k>0$ and $y\in[0,1]$, we have  
\be
\tilde Z^{(t)} & \leq & (1+\epsilon_Z) Z^{(t)} =  (1+\epsilon_Z)\displaystyle\sum_{i=1}^n w_{i}^{(t)} e^{\eta \frac{\rho^{(t)}_i}{\tilde I^{(t)}}} \leq (1+\epsilon_Z)\displaystyle\sum_{i=1}^n w^{(t)}_{i} \left(1-\left(1-e^{\frac{\eta}{\tilde I}}\right)\rho^{(t)}_i\right)\\
& = & (1+\epsilon_Z)\left(1 - \left(1- e^{\frac{\eta}{\tilde I^{(t)}}}\right)\tilde I^{(t)}\right).
\ee
Using the fact that $\log\left(1-c\left(1-e^y\right)\right)\leq c y + \frac{y^2}{8}$ for all $c\in[0, 1]$ and $y\in\mathbb{R}$, $\log \tilde Z_t$ can be bounded by
\be 
\log \tilde Z^{(t)} 
& = & \log\left(1+\epsilon_Z\right) + \log\left(1 - \left(1- e^{\frac{\eta}{\tilde I^{(t)}}}\right)\tilde I^{(t)}\right)
\leq \epsilon_Z + \eta + \frac{\eta^2}{8(\tilde I^{(t)})^2}\\
& \leq & \epsilon_Z + \eta + \frac{\eta^2}{8 (I^{(t)})^2} \left(1+\frac{4}{3}\epsilon_I\right)^2 
\leq \epsilon_Z + \eta + \frac{2\eta^2}{9 (I^{(t)})^2} \leq \epsilon_Z + \eta + \frac{2\eta^2}{9r^2_{\min}},
\ee
which is true as $\log(1+\epsilon_Z)\leq \epsilon_Z$ since $\epsilon_Z\geq 0$ and using $\rho^{(t)}_i\geq r_{\min}$. Combining with Eq.~(\ref{change}), we have   
\be
\Delta_t & \leq & -\eta\left(1-\frac{4}{3}\epsilon_I\right)\frac{u\cdot\rho^{(t)}}{I^{(t)}} + \epsilon_Z + \eta + \frac{2\eta^2}{9 r^2_{\min}}\\
& = & \eta\left(1 - \left(1-\frac{4}{3}\epsilon_I\right)\frac{u\cdot\rho^{(t)}}{I^{(t)}}\right) + \epsilon_Z +  \frac{2\eta^2}{9 r^2_{\min}}\\
& \leq & -\eta \log\left(\frac{u\cdot\rho^{(t)}}{I^{(t)}}\right) + \frac{4\eta}{3}\epsilon_I \frac{u\cdot\rho^{(t)}}{I^{(t)}} + \epsilon_Z  + \frac{2\eta^2}{9 r^2_{\min}},\\
& \leq & -\eta \log\left(\frac{u\cdot\rho^{(t)}}{I^{(t)}}\right) + \epsilon_I \frac{4\eta}{3r_{\min}} + \epsilon_Z + \frac{2\eta^2}{9 r^2_{\min}},
\ee
where we use the fact that $1-e^x\leq -x$ for all $x$. Using a telescoping sum over $t\in[T]$, we have 
\be
-D\left(u\Vert w^{(1)}\right) & \leq & D\left(u\lVert w^{(T+1)}\right) - D\left(u\Vert w^{(1)}\right)\\
& \leq & \eta\displaystyle\sum_{t=1}^T \left(\left(w^{(t)}\cdot \rho^{(t)}\right) - \log\left(u\cdot \rho^{(t)}\right) \right) + \epsilon_I \frac{4\eta T}{3r_{\min}} + T \epsilon_Z +   \frac{2\eta^2 T}{9r_{\min}^2}.
\ee
Rearranging the terms, we obtain 
\be
\displaystyle\sum_{t=1}^T \log\left(w^{(t)}\cdot \rho^{(t)}\right) & \geq & \sum_{t=1}^T \log\left(u\cdot \rho^{(t)}\right) - \frac{D(u\lVert w^{(1)})}{\eta} - \epsilon_I \frac{4 T}{3r_{\min}} - \frac{T\epsilon_Z}{\eta} -  \frac{2\eta T}{9r_{\min}^2}.
\ee
Since we let $w^{(t1)} = \left(\frac{1}{n}, \cdots, \frac{1}{n}\right)$, we have $D(u\lVert w^{(1)})\leq \log n$. Setting
\be
\epsilon_I = \frac{3\eta}{4 r_{\min}}, \hspace{1cm} \epsilon_Z =0,
\ee
gives 
\be
\displaystyle\sum_{t=1}^T \log\left(w^{(t)}\cdot \rho^{(t)}\right)\geq \sum_{t=1}^T \log\left(u\cdot \rho^{(t)}\right) - \frac{3\sqrt{2T\log n}}{r_{\min}}
\ee
Setting 
\be
\epsilon_I = \frac{3\eta}{4 r_{\min}}, \hspace{1cm} \epsilon_Z = \frac{\eta^2}{r^2_{\min}},
\ee
gives
\be
\displaystyle\sum_{t=1}^T \log\left(w^{(t)}\cdot \rho^{(t)}\right)\geq \sum_{t=1}^T \log\left(u\cdot \rho^{(t)}\right) - \frac{5\sqrt{2T\log n}}{r_{\min}}
\ee
Thus, we obtain the desired bounds. 
\end{proof}

\subsection{Classically-sampled inner product}
We show an algorithm where we classically sample the inner product. This algorithm does not offer any reduction in the run time, and the transaction cost is the same as in Algorithm~\ref{algo_helmbold_sampling}. We present the algorithm to provide a gradual transition to the quantum algorithm, as the correctness analysis will be similar for the quantum algorithm.
First, 
restate a lemma on classically-sampled inner products as follows. 
We would like to highlight that for all $t\in[T]$, the strategies $w^{(t)}$ are the same for Algorithm~\ref{algo_helmbold} and  Algorithm~\ref{algo_helmbold_sampling}, but, due to this erroneous update, are different for the following Algorithm~\ref{algo_helmbold_sampling_inner} and Algorithm~\ref{algo}. 
The following lemma estimates inner products with relative error. As in our portfolio setting, the vector $x$ has a lower bound for its entries.
\begin{lem}[Inner product estimation]\label{IP_est} 
Let $\delta\in(0,1)$. Given query access to $x\in[x_{\min}, 1]^n$ and $\ell_1$-sampling access to a probability vector $p\in[0, 1]^n$, we can determine, with success probability at least  $1-\delta$, the inner product $\alpha := p\cdot x \in [x_{\min},1]$ to multiplicative error $\epsilon_I \leq x_{\min}$, with $O\left(\frac{1}{\epsilon_I^2 x_{\min}}\log\frac{1}{\delta}\right)$  queries and samples, and $\tilde{\mathcal{O}}\left(\frac{1}{\epsilon_I^2 x_{\min}}\log\frac{1}{\delta}\right)$
time complexity.
\end{lem}
\begin{proof}
Consider the additive version of this lemma as given in Reference~\cite{tang2019quantum, Rebentrost2021,chia2020sampling, chia2018quantum}, adapted to the $\ell_1$ case: Let $X$ be a random variable with outcome $x_j$ with probability $p_j$.
Note that $\mathbb{E}[X] = \displaystyle\sum_j p_j x_j = \alpha$ and $Var(X) \leq \displaystyle\sum_j x_j^2 p_j \leq \alpha$. Apply the median-of-means method \cite{chia2020sampling} on $\frac{27}{\epsilon^2}\log \frac{1}{\delta}$ samples of $X$
to be within $\epsilon \sqrt{Var(X)} \leq \epsilon \sqrt{\alpha}$ of $p \cdot x$
with probability at least $1-\frac{\delta}{2}$ using $O\left(\frac{1}{\epsilon^2} \log \frac{1}{\delta}\right)$ queries.

For the multiplicative estimation, run the above algorithm with the precision parameter being set to $\epsilon = \epsilon_I$. We obtain an estimate  $\widetilde{\alpha}_1$ of the inner product $\alpha = p\cdot x$ such that $|\alpha - \widetilde{\alpha}_1| \leq \epsilon_I \sqrt{\alpha}$, Then, re-run the algorithm with precision $\epsilon = \epsilon_I \sqrt{\widetilde{\alpha}_1}/2$. We will in turn obtain an estimate $\widetilde{\alpha}_2$ such that
\be
|\alpha - \widetilde{\alpha}_2| & \leq & \frac{\epsilon_I  \sqrt{\widetilde{\alpha}_1}}{2} \sqrt{\alpha} \leq \frac{\epsilon_I\sqrt{\alpha + \epsilon_I}}{2} \sqrt{\alpha} \leq \epsilon_I\alpha.
\ee
This costs $O\left(\frac{1}{\epsilon_I^2 x_{\min}} \log \frac{1}{\delta}\right)$ queries to obtain the desired guarantee. 
\end{proof}

\begin{algorithm}
\caption{Approximate sampling-based Online Portfolio Optimization Algorithm}
\label{algo_helmbold_sampling_inner}
\begin{algorithmic}[1]
\Require $n$, $s$, $\eta$, $T$, $C$, $\delta$.
\State Initialize $w^{(1)}=(\frac{1}{n}, \cdots, \frac{1}{n})\in\mathbb{R}^n$.
\For {$t=1$ to $T$}
\State Prepare sampling data structure for $w^{(t)}$ using Fact \ref{sampling_data_structure}. 
\State Sample $i^{(t)}_1, \cdots, i^{(t)}_s$ from $w^{(t)}$.
\State Invest the amount $1/s$ in each asset $i^{(t)}_1, \cdots, i^{(t)}_s$ at cost $C$ each. 
\State Wait until end of day. 
\State Receive price relative vector $\rho^{(t)}$.
\State $\tilde I^{(t)} \gets$ estimate inner product $w^{(t)}\cdot\rho^{(t)}$ via Lemma~\ref{IP_est} with relative error $\epsilon_I=\frac{3\eta}{4 r_{\min}}$ and success probability $1-\frac{\delta}{T}$.
\For {$i=1$ to $n$}
\State $w^{(t+1)}_{i}\gets EEG(\eta,  w^{(t)},  \rho^{(t)})_i$, with $\tilde I^{(t)}$ and exact norm computation.
\EndFor 
\EndFor
\Ensure $LS^C_{\text{samp}}\coloneqq \frac{1}{T}\displaystyle\sum_{t=1}^T \log \left(\frac{1}{s}\displaystyle\sum_{\ell=1}^s \rho^{(t)}_{i_\ell^{(t)}}  \right)$. 
\end{algorithmic}
\end{algorithm}
Thm.~\ref{algo_helmbold_sampling_inner} implies the corollary below. 
\begin{corollary}[Guarantee and run time of Algorithm~\ref{EG_Helmbound_bound}]\label{apprx_Helmbund_regretbound}
Let $\delta \in (0,1/3)$.
Algorithm~\ref{EG_Helmbound_bound} with $\eta = 2r_{\min}\sqrt{\frac{2\log n}{T}}$, achieves
\be\label{regret_LSsampC}
LS^* - LS^C_{\text{samp}} & \leq & \frac{8}{r_{\min}}\sqrt{\frac{\log n}{2T}}.
\ee
with success probability at least $1-3\delta$
in $O\left(T n + \frac{T^2}{r_{\min}}\log\frac{1}{\delta} + T^2\log\frac{T}{\delta}\right)$ time and incurs a transaction cost of $O\left(T^2C\log\frac{T}{\delta}\right)$.
\end{corollary}
\begin{proof}
For the guarantee, we use Thm.~\ref{regret-bound_helmold_sampling}
and Thm~\ref{EG_Helmbound_bound}, and omit more detailed steps here.
Aside from the inner product sampling, the run time and transaction cost is the same as in Thm.~\ref{regret-bound_helmold_sampling}. 
A single inner product sampling to accuracy $\epsilon_I = \frac{3}{2} \sqrt{\frac{2\log n}{T}}$ takes time $O\left(\frac{T}{r_{\min}}\log\frac{1}{\delta}\right)$, and is performed $T$ times in the algorithm. 
A union bound of all steps in the algorithm succeeding and the Hoeffding bound leads to the stated total success probability. 
\end{proof}
\subsection{The quantum online portfolio optimization algorithm}

We now present our quantum online portfolio optimization algorithm and its analysis. We change the input assumption to a natural quantum extension of the classical input.
The correctness guarantee essentially follows from Theorem \ref{EG_Helmbound_bound}. We obtain a quadratic speedup in the run time compared to the classical algorithm.
Our quantum online portfolio optimization algorithm makes use of the following procedures: quantum state preparation, norm estimation, and inner product estimation. We also employ a multi-sampling algorithm \cite{Hamoudi2019} as our subroutine to allow us to sample $s$ elements from a collection of $n$ elements in about $\sqrt{sn}$ time instead of  about $s\sqrt{n}$ time. 
Before we present the main algorithm, we will introduce these quantum subroutines. 

In the quantum setting, instead of classical access to the price relatives we assume quantum access to the price relatives. The online nature of the problem is given by the fact that we obtain these oracles at the different times.    
\begin{datainput}[Online gain oracles] \label{datainput}  
Let $\rho^{(1)}, \cdots, \rho^{(T)}$ be price relatives with $\max_{i\in[n]} \rho^{(t)}_i=1$ and $\rho^{(t)}_i\geq r_{\min}>0$ for all $i\in[n], t\in[T]$.
Define the unitary $U_{\rho^{(t )}}$ operating on $O(\log n)$ quantum bits such that for all $j\in[n]$,  $U_{\rho^{(t)}}\ket{j}\ket{\vec{0}}=\ket{j}\ket{\rho_{j}^{(t)}}$. 
At time $t\in[T]$ (end of day), assume access to unitaries $U_{\rho^{(1)}}, \cdots, U_{\rho^{(t)}}$.
\end{datainput}
The update rule Eq.~(\ref{update_rule}) can be rewritten in terms of a sum over all previous price relatives and inner products, as the following observation shows.
\begin{fact} Let $w^{(1)} = \left(\frac{1}{n}, \cdots, \frac{1}{n}\right)$. Given the update rule Eq.~(\ref{update_rule}), we can express, for $t\in[T], i\in[n]$,
\be
w^{(t+1)}_i & = & \frac{\exp\left(\eta\sum_{t^\prime=1}^t \frac{\rho^{(t)}_i}{w^{(t^\prime)}\cdot\rho^{(t^\prime)}}\right)}{\sum_{j=1}^n \exp\left(\eta\sum_{t^\prime=1}^t \frac{\rho^{(t^\prime)}_j}{w^{(t^\prime)}\cdot \rho^{(t^\prime)}}\right)}.
\ee
\end{fact}
Given Data Input \ref{datainput}, the following computations can be performed in superposition of the index $i\in[n]$ for the assets. Similar unitaries were studied in, e.g., References~\cite{vedral1996quantum,Li2019, chakrabarti2021threshold,rebentrost2021quantum}.
\begin{lem}\label{sum_ratios}
Let $t\leq T$.
Let there be given the set of unitaries $U_{\rho^{(t^\prime)}}$ for $t^\prime\in[t-1]$ as in Data Input \ref{datainput}, a vector $\tilde{I}\in\mathbb{R}^{t-1}$, and some reals $\eta,a \in\mathbb{R}$. 
There exists unitary operators performing the following computations:
\be\label{divide}
\ket{i}\ket{\bar{0}} & \rightarrow & \ket{i}\ket{\displaystyle\sum_{t^\prime=1}^{t-1}\frac{\rho^{(t^\prime)}_i}{\tilde{I}^{(t^\prime)}}},\quad
\ket{i}\ket{\bar{0}}  \rightarrow  \ket{i}\ket{q^{(t)}_i}, \quad
\ket{i}\ket{\bar{0}}  \rightarrow  \ket{i}\ket{\frac{q^{(t)}_i}{a}},
\ee
where $q^{(t)}_i = \exp\left(\eta\displaystyle\sum_{t^\prime=1}^{t-1} \frac{\rho^{(t^\prime)}_i}{\tilde{I}^{(t^\prime)}}\right)$
to sufficient numerical precision. These computations take $O(T)$ queries and to the data input and requires $O(T+\log n)$ qubits and quantum gates. 
\end{lem}
\begin{proof}
With a computational register involving $O(T)$ ancilla qubits for the gains and the ratios, perform 
\be
\ket{j}\ket{\bar{0}} 
& \rightarrow & \ket{j}\ket{\rho^{(1)}_j}\cdots \ket{\rho^{(t-1)}_j}
\ket{\bar{0}}\\
& \rightarrow & \ket{j}\ket{\rho^{(1)}_j}\cdots \ket{\rho^{(t-1)}_j}
\ket{\frac{\rho^{(1)}_j}{\tilde{I}^{(1)}}}\cdots \ket{\frac{\rho^{(t-1)}_j}{\tilde{I}^{(t-1)}}}\ket{\bar{0}}\\
& \rightarrow & \ket{j}\ket{\rho^{(1)}_j}\cdots \ket{\rho^{(t-1)}_j}
\ket{\frac{\rho^{(1)}_j}{\tilde{I}^{(1)}1}}\cdots \ket{\frac{\rho^{(t-1)}_j}{\tilde{I}^{(t-1)}}}\ket{\displaystyle\sum_{t^\prime=1}^{t-1} \frac{\rho^{(t^\prime)}_j}{\tilde{I}^{(t^\prime)}}},
\ee
to sufficient accuracy using the oracles and quantum circuits for basic arithmetic operations. Uncomputing the intermediate registers with additional queries gives us the desired result. In addition, computations of $\ket{i}\ket{q^{(t)}_i}$ and $\ket{i}\ket{\frac{q^{(t)}_i}{a}}$ can be achieved using the quantum circuits for basic arithmetic operations.
\end{proof}

We restate the quantum state preparation, norm and inner product estimation procedure from References~\cite{Brassard2002, Li2019, VanApeldoorn2019b, Hamoudi2019, Rebentrost2021} for the convenience of the reader.
\begin{lem}[Quantum state preparation and norm estimation]\label{lemma5}
Given a vector $v\in[0, 1]^n$ with $\max_j v_j=1$ and quantum access to $v$. Then:
\begin{enumerate}[(i)]
\item  \label{lemma5(ii)} Let $\epsilon>0$ and $\delta\in(0,1)$. There exists a quantum algorithm that outputs an estimate $\tilde{z}$ of the $\ell_1$-norm $\lVert v\rVert_1$ of $v$, such that $|\lVert v\rVert_1 - \tilde{z}|\leq \epsilon\lVert v\rVert_1$, with probability at least $1-\delta$. The algorithm uses $O\left(\frac{\sqrt{n}}\epsilon{\log{\left(\frac{1}{\delta}\right)}}\right)$ queries and $\tilde{O}\left(\frac{\sqrt{n}}\epsilon{\log{(\frac{1}{\delta})}}\right)$ gates.
\item \label{lemma5(iii)} Let $\zeta \in (0,1/2]$ and $\tilde z>0$ be given such that $\vert \tilde z - \lVert v\rVert_1 \vert \leq \zeta \lVert v\rVert_1 $. Let $\delta\in(0,1)$. An approximation $\ket{\tilde{p}}=\displaystyle\sum_{j=1}^n\sqrt{\tilde{p}_j}\ket{j}$ to the state $\ket{v}=\displaystyle\sum_{j=1}^n \sqrt{\frac{v_j}{\lVert v\rVert_1}}\ket{j}$ can be prepared with probability $1-\delta$ using $O(\sqrt{n}\log{\frac{1}{\delta}})$ calls to the unitary of (i) and $\tilde{O}\left(\sqrt{n}\log{(\frac{1}{\delta})}\right)$ gates. The approximation in $\ell_1$-norm of the probabilities is $\big\lVert \tilde{p}-\frac{v}{\lVert v\rVert_1}\big\rVert_1\leq 2\zeta$.
\end{enumerate}
\end{lem}
\begin{proof} 
There exists a unitary operator  that prepares the state $\frac{1}{\sqrt{n}}\displaystyle\sum_{j=1}^n\ket{j}\big(\sqrt{v_j}\ket{0}+\sqrt{1-v_j}\ket{1}\big)$ with two queries to the vector, a controlled rotation, and $O\big(\log n\big)$ gates \cite{Hamoudi2019}. 
Using this unitary, the proof of part (\ref{lemma5(ii)}) is as in References~\cite{Brassard2002,Li2019,Rebentrost2021}. For part (\ref{lemma5(iii)}), note that for $i\in[n]$, we have that $\tilde{p}_i = \frac{v_i}{\tilde{z}}$ and $\lVert v\rVert_1 - \tilde{z} \leq \lvert \lVert v\rVert_1 - \tilde{z}\rvert \leq \zeta\lVert v\rVert_1$. Hence, $\frac{1}{\tilde{z}}\leq \frac{1}{\lVert v\rVert_1(1-\zeta)}$. Then, $\left\lVert \tilde{p} - \frac{v}{\lVert v\rVert_1}\right\rVert_1
= \displaystyle\sum_{i=1}^n\left\lvert \frac{v_i\lVert v\rVert_1 - v_i\tilde{z}}{\tilde{z}\lVert v\rVert_1}\right\rvert
\leq \frac{\zeta}{1-\zeta} \displaystyle\sum_{i=1}^n \frac{ v_i}{\lVert v\rVert_1} 
\leq 2\zeta.$
\end{proof}
Next, we show a lemma on the estimation of inner products. Using quantum maximum finding, we are able to use the norm computation of Lemma~\ref{lemma5} for outputting the inner product.
\begin{lem}[Quantum inner production estimation with relative accuracy \cite{Rebentrost2021}]\label{QIP}
Let $\epsilon, \delta\in(0,1)$ and given quantum access to a non-zero vector $u \in [u_{\min},1]^n$ and a probability vector $v \in [0, 1]^n$ such that $\displaystyle\sum_{i=1}^n v_i = 1$. Then, an estimate $\widetilde{IP}$ for the inner product can be obtained such that $\left\vert \widetilde{IP} - u\cdot v\right\vert \leq\epsilon\ u\cdot v$ with success probability $1-\delta$. This requires $O\left(\frac{\sqrt{n}}{\epsilon\sqrt{u_{\min}}}\log{(\frac{1}{\delta})}\right)$ queries and 
$\tilde{O}\left(\frac{\sqrt{n}}{\epsilon\sqrt{u_{\min}}}\log{(\frac{1}{\delta})}\right)$ quantum gates. 
\end{lem} 
\begin{proof}
Prepare the state $\ket{\psi^\prime} = \displaystyle\sum_{i=1}^n \sqrt{v_i}\ket{u_i}\ket{i}$. Next, a controlled-rotation obtains 
\be
\ket{\psi} = \displaystyle\sum_{i=1}^n \sqrt{v_i}\ket{u_i}\ket{i}\left(\sqrt{u_i}\ket{0} + \sqrt{1 - u_i}\ket{1}\right).
\ee
Define unitaries $U = 2\ket{\psi}\bra{\psi} - I$ and $V = I - 2P$ for some projector $P$. Then, we have 
\be
\braket{\psi}{P\vert\psi} = \displaystyle\sum_{i=1}^n v_i u_i = u\cdot v
\ee
Amplitude estimation \cite{Brassard2002} allows us to to estimate $a  =\sqrt{n}(u\cdot v)$ to accuracy $\left\vert a - \tilde a\right\vert\leq 2\pi\frac{\sqrt{a(1-a)}}{C} + \frac{\pi^2}{C^2}$ with probability at least $\frac{8}{\pi^2}$ using $C$ applications of $U $ and $V$. Setting $C = \frac{6\pi\sqrt{n}}{\epsilon\sqrt{u_{\min}}}$, we obtain 
\be
\left\vert a - \tilde a\right\vert
& \leq & \frac{\pi}{C}\left(2\sqrt{a} + \frac{\pi}{C}\right)\\
& = & \frac{\epsilon \sqrt{u_{\min}}}{6\sqrt{n}}\left(2\sqrt{a} + \frac{\epsilon \sqrt{u_{\min}}}{6\sqrt{n}}\right)\\
(\text{since $u_{\min}\leq a$}) & \leq & \frac{\epsilon a}{6\sqrt{n}}\left(2 + \frac{\epsilon}{6\sqrt{n}}\right)\\
& \leq & \frac{\epsilon a}{2\sqrt{n}}\\
& \leq & \epsilon a\\
\ee
We then repeat the above procedure for $O\left(\log \left(\frac{1}{\delta}\right)\right)$ times to boost the success probability to $1-\delta$. The run time is $O\left(\frac{\sqrt{n}}{\epsilon\sqrt{u_{\min}}}\log\left(\frac{1}{\delta}\right)\right)$.
\end{proof}

In the quantum setting, multi-sampling can be done using the quantum algorithm from Ref.~\cite{Hamoudi2019}. Lemma~\ref{multi-sampling_parameters} uses  quantum maximum finding, quantum norm estimation and Grover's search to find the inputs to the quantum multi-sampling algorithm.
\begin{lem}[\cite{Hamoudi2019}]\label{multi-sampling_parameters}
Let $\epsilon, \delta\in (0, 1)$, $1<s<n$ be an integer and $p\in\mathbb{R}^n$ be a non-zero vector. For any set $S\subseteq[n]$, denote as $p_S\in\mathbb R^{\lvert S\vert}$ the subvector $(p_i)_{i\in S}$ of $p$ that consists of the values at coordinates $i\in S$. There exists a quantum  algorithm that takes $\epsilon, \delta, s, p$ as inputs and returns a real $\Gamma>0$, a set $W\subseteq [n]$ and a value $V$ that satisfy the following conditions: 
\begin{enumerate}[(i)]
\item $\Gamma\geq\lVert p_W\rVert_1$
\item $\vert  \Gamma - \lVert p\rVert_1\vert\leq \min\{1/\sqrt{s},\epsilon\}\lVert p\rVert_1$
\item $W = \{i\in[n]: \vert p_i\vert\geq\Gamma/s\}$ 
\item $V = \lVert p_{[n]\backslash W}\rVert_\infty$
\end{enumerate}
with probability $1-\delta$ in $O\left((\sqrt{sn} + \sqrt{sn}/\epsilon)\log\left(1/\delta\right)\right)$. 
\end{lem}

The following quantum multi-sampling algorithm allows us to achieves a quadratic speedup in the sampling run time by using amplitude amplification.
\begin{lem}[Quantum multi-sampling algorithm \cite{Hamoudi2019}]\label{multi_sampling_algo}
Let $1<s<n$ be an integer, $0<\delta<1$ be a real number and $p\in\mathbb{R}^n$ be a non-zero vector. Given $\Gamma>0$, a set $W\in[n]$ such that $\lVert p_W\rVert_1\leq \Gamma$, a value $V = \lVert p_{[n]\backslash W}\rVert_\infty$ and  query access to $p$, there exists a quantum algorithm that output $s$ independent samples from $p$ in expected time $O\left(\sqrt{sn}\log\left(\frac{1}{\delta}\right)\right)$ with probability $1-\delta$. 
\end{lem}

We now present our quantum online portfolio optimization Algorithm \ref{algo} and its analysis. 
The correctness guarantee uses Theorem \ref{EG_Helmbound_bound}. Using Lemma \ref{QIP} and the other quantum subroutines we obtain a quadratic speedup in the run time compared to Alg.~\ref{algo_helmbold_sampling_inner}.
The following theorem gives our main result for the regret and the run time of Algorithm~\ref{algo}.
\begin{algorithm}
\caption{Quantum Online Portfolio Optimization Algorithm}
\label{algo}
\begin{algorithmic}[1]
\Require $n$, $s$, $\eta$, $T$, $C$, $\delta$
\State Initialize  $U_{\rho^{(0)}} = \mathbb{I}$, $\tilde{I} = 1$ and $q^{(1)}=(1, \cdots, 1)$, $U_{q^{(1)}} = \mathbb{I}$. 
\For {$t=1$ to $T$}
\State $q_{\max}\gets$ Find the largest element of $q^{(t)}$ using $U_{q^{(t)}}$ and quantum maximum finding \cite{Durr1996} with success probability $1-\frac{\delta}{4T}$. 
\State $\tilde Z^{(t)} \gets$ Estimate the norm of $\frac{q^{(t)}}{q_{\max}}$ using $U_{q^{(t)}}$, $q_{\max}$, Lemma~\ref{sum_ratios}, and  Lemma~\ref{lemma5}(\ref{lemma5(ii)}), with relative error $\epsilon_Z =\frac{\eta^2}{r^2_{\min}}$ and success probability $1-\frac{\delta}{4T}$.
\State $U_{w^{(t)}}\gets $ Prepare quantum circuit for approximating $\ket{\tilde{w}^{(t)}}$ of the quantum state $\ket{w^{(t)}}$, where $w^{(t)} = \frac{q^{(t)}}{\lVert q^{(t)}\rVert_1}$, using $U_{q^{(t)}}$, $q_{\max}, \tilde{Z}^{(t)}$ 
, Lemma~\ref{sum_ratios}, Lemma~\ref{lemma5}(\ref{lemma5(iii)}), with success probability $1-\frac{\delta}{4T}$.  
\State $\Gamma, W, V\gets$ Determine using Lem.~\ref{multi-sampling_parameters} applied to $\tilde w^{(t)}$ with probability $1-\frac{\delta}{4 T}$. 
\State $i^{(t)}_1, \cdots, i^{(t)}_s$ Perform multi-sampling using $\Gamma, W, V$ and Lemma~\ref{multi_sampling_algo} with probability $1-\frac{\delta}{4 T}$. 
\State Invest the amount $1/s$ in each asset $i^{(t)}_1, \cdots, i^{(t)}_s$ at cost $C$ each. 
\State Wait until end of  day.
\State Receive price relative oracle $U_{\rho^{(t)}}$.  
\State $\rho^{(t)}_{j^{(t)}}\gets$ Query $U_{\rho^{(t)}}$ with $\ket{j^{(t)}}\ket{0}$.
\State $\widetilde {I}^{(t)} \gets$ Estimate $\tilde w^{(t)}\cdot \rho^{(t)}$ using $U_{w^{(t)}}$, $U_{\rho^{(t)}}$, and Lemma \ref{QIP}, with relative error $\epsilon_I = \frac{3\eta}{4 r_{\min}}$ and success probability $1-\frac{\delta}{4T}$.
\State $U_{q^{t+1}}\gets$ Prepare quantum circuit to compute $q^{(t+1)} =  \exp\left(\eta\displaystyle\sum_{t^\prime=1}^{t} \frac{\rho^{(t^\prime)}}{\tilde{I}^{(t^\prime)}}\right)$ using $\{ U_{\rho^{(t')}}\}_{t'=1}^t$ and Lemma \ref{sum_ratios}.
\EndFor
\Ensure $LS_{\text{samp}}^Q\coloneqq \frac{1}{T}\displaystyle\sum_{t=1}^T \log\left(\frac{1}{s}\displaystyle\sum_{\ell=1}^s \rho^{(t)}_{i^{(t)}_{\ell}}\right)$. 
\end{algorithmic}
\end{algorithm}
\begin{theorem}[Quantum online portfolio optimization] \label{thmQuantum}
Let $\delta\in(0,1/3)$.
Algorithm~\ref{algo} with $\eta = 2r_{\min}\sqrt\frac{2\log n}{T}$, outputs $LS_{\text{samp}}^Q$ with the regret bound
\be
LS^* - LS^Q_{\text{samp}} \leq \frac{12}{r_{\min}}\sqrt{\frac{\log n}{2T}},
\ee
with success probability at least $1-3\delta$. The total run time is $O\left(T^3\sqrt{\frac{n}{r_{\min}}}\log^{1.5}\left(\frac{1}{\delta}\right)\right)$ and the transaction cost is $O\left(T^2C\log\frac{T}{\delta}\right)$.
\end{theorem}

\begin{proof}
Condition the following argument on all probabilistic steps of the algorithm succeeding, which occurs with probability $1-\delta$ from the union bound. At each time step $t\in[T]$, the quantum algorithm produces a portfolio vector $\tilde{w}^{(t)}$. Similar to the proof of Thm.~\ref{regret-bound_helmold_sampling}, we define the random variable $Y^{(t)} = \frac{1}{s}\displaystyle\sum_{\ell=1}^s \rho^{(t)}_{i^{(t)}_\ell}$ with probability $\tilde w_{i^{(t)}_\ell}$.
Then, the expected value of the random variable is 
\be\label{Jensen}
\mathbb E\left[Y^{(t)}\right] = \tilde w^{(t)}\cdot\rho^{(t)} =: \widetilde{LS}. 
\ee
Similar to the analysis of Thm.~\ref{regret-bound_helmold_sampling}, we obtain
\be\label{EsampLStilde2}
\widetilde{LS} - LS^Q_{\text{samp}}\leq \frac{1}{r_{\min}}\sqrt{\frac{1}{2T}}
\ee
with probability at least $1-2\delta$. Using Theorem~\ref{EG_Helmbound_bound} and  Eq.~(\ref{EsampLStilde2}), the regret is bounded by
\be
LS^* - LS^Q_{\text{samp}} 
& = & LS^* - \widetilde{LS}  + \widetilde{LS} - LS^Q_{\text{samp}} \\
& \leq & \frac{9}{ r_{\min}}\sqrt{\frac{2\log n}{T}} + \frac{1}{r_{\min}}\sqrt{\frac{1}{2T}}\\
& \leq & \left(\frac{12}{r_{\min}}\right)\sqrt{\frac{\log n}{2T}},
\ee
with success probability at least $1-2\delta$. By the union bound, the total success probability is at least $1-3\delta$ by taking into account the success probability of the algorithm. 
For the run time, 
consider that $U_{q^{(t)}}$ costs $O(T+\log n)$ by Lemma \ref{sum_ratios}. Using the values for $\eta$,  $\epsilon_I$ and $\epsilon_Z$, the total run time is 
\be
& & O(T(\text{quantum maximum finding} + \text{quantum norm estimation} \\
& + & \text{quantum state preparation} + \text{quantum inner product estimation}\\
& + & \text{quantum multi-sampling})) \\
& \subseteq & \tilde O\left(T^2\left( \sqrt{n}\log\left(\frac{1}{\delta}\right) + \frac{\sqrt{n}}{\epsilon_Z}\log\left(\frac{1}{\delta}\right) + \sqrt{n}\log\left(\frac{1}{\delta}\right) + \frac{\sqrt{n}}{\epsilon_I\sqrt{r_{\min}}}\log\left(\frac{1}{\delta}\right)
+ \sqrt{sn}\log\left(\frac{1}{\delta}\right)\right) \right) \nonumber\\
&\subseteq  & \tilde O\left(T^2\sqrt{n}\left( 2+ \sqrt{\frac{T}{r_{\min}\log n}} + \frac{T}{\log n} + \sqrt{T}(1-r_{\min})\sqrt{\log\left(\frac{T}{\delta}\right)}\right)\log\left( \frac{1}{\delta}\right)\right)\\
& \subseteq & \tilde O\left(T^3\sqrt{\frac{n}{r_{\min}}}\log\left(\frac{1}{\delta}\right)\sqrt{\log\left(\frac{T}{\delta}\right)}\right)\\
& \subseteq & \tilde O\left(T^3\sqrt{\frac{n}{r_{\min}}}\log^{1.5}\left(\frac{1}{\delta}\right)\right).
\ee
\end{proof}

\section{Discussion and conclusion}
The online setting is more general than the offline setting in that it allows for the input to be given sequentially, where the sequence could be chosen adversarially. The adversarial property implies that the inputs could be selected with the knowledge of the present state of the algorithm, say,  with the knowledge of our portfolio vector. The regret bounds hold nevertheless, also in the quantum setting. This online setting is rather natural for certain portfolio optimization situations, where the investment strategy can be inferred by other market participants from transaction data.  Online algorithms in the portfolio optimization context have been studied in practice in References~\cite{Helmbold1998,wang2020kernel,li2018transaction, khedmati2020online}.

We have devised a quantum online portfolio optimization algorithm that runs in time $\tilde O\left(\frac{T^3\sqrt{n}}{r_{\min}}\log\left(\frac{1}{\delta}\right)\right)$ and has transaction cost that is independent on the number of assets. Our quantum algorithm achieves a slightly worse (by a constant factor) regret bound, but is more space efficient as compared to its classical counterpart \cite{Helmbold1998}, not considering the space requirement for the input oracles. The classical online portfolio optimization algorithm by Reference~\cite{Helmbold1998}  uses linear (in terms of the number of assets) time and space to update the portfolio vector in every iteration. In our quantum algorithm, we leverage on the fact that the portfolio vectors can be computed efficiently via unitaries that perform arithmetic operations to save on the space/memory of the algorithm. Nevertheless, the practical implementation of the price relative oracles appears to be a bottleneck for this algorithm. In particular, building a QRAM for each of the oracles requires $O(Tn\log n)$ time and $O(Tn)$ space.

We note that in both the classical and quantum settings, we know the identity of the assets that we are investing in after we have sampled the corresponding indices. In the quantum setting, we do not perform full tomography of the portfolio vector and hence do not incur the corresponding cost. 
We provide a comment on the online setting in contrast to the standard Markovitz mean-variance portfolio optimization. The online setting takes into account variance and covariance of the asset prices implicitly via the time series of prices relatives.
The algorithms are favourable when the asset prices have bounded relative volatility \cite{Helmbold1998a}, because they assume knowledge of the upper and lower bounds on the price relatives. Since $r_{\min}\leq \rho^{(t)}_j\leq 1$, 
the variance of each entry of the price relatives and the covariance between entries are upper bounded by $\frac{(1-r_{\min})^2}{4}$ by Fact~\ref{Popoviciu} and hence the volatility (standard deviation) is $\frac{1-r_{\min}}{2}$. 
Thus, the maximum volatility of the market is taken into account by the bounds on the price relatives.

In our setting, the transaction cost was taken to be a constant for each investment, independent of the size of the investment. This models the fact that for each investment some fixed amount of work has to be performed, e.g., the communication of the trade between counterparties and the transfer of the asset. This type of transaction cost serves to illustrate the benefits of the sampling algorithm over the standard algorithm. 
For future work, one can consider imposing additional constraints on the portfolio optimization problem. For instance, a common optimization is to minimize transaction cost via including a term $\Vert w^{(t)} - w^{(t-1)}\Vert_1$ in the portfolio optimization problem or consider portfolio optimization in the robust setting, where the parameters belong to an uncertainty set. Various flavours of robustness such as constraint, objective and relative robustness in conjunction with different types of uncertainty sets \cite{cornuejols2006optimization}  are also worth investigating.  

\section{Acknowledgements}
Research at CQT is funded by the National Research Foundation, the Prime Minister’s Office, and the Ministry of Education, Singapore under the Research Centres of Excellence programme’s research grant R-710-000-012-135. We also acknowledge funding from the Quantum Engineering Program (QEP 2.0) under grant NRF2021-QEP2-02-P05.

\bibliographystyle{unsrt}
\bibliography{QOPO_2023}{}

\begin{appendices}
\section{Useful inequality and lemma}
\begin{fact}[Hoeffding's inequality]\label{Hoeffding}
Let $(\Omega,\Sigma,\mathbb{P})$ be a probability space, 
and consider random variables $Y_1, \cdots, Y_n$, where $Y_i : \Omega\to[c_i,d_i]$ for $i\in[n]$. Let $Z = \displaystyle\sum_{i=1}^n Y_i$. Then, for all $a>0$, 
\be
\mathbb{P}\left[\left\lvert Z - \mathbb{E}\left[Z\right]\right\rvert\geq a\right] \leq 2\exp\left(-\frac{2a^2}{\sum_{i=1}^n (d_i - c_i)^2}\right).
\ee
\end{fact}

\begin{fact}[Popoviciu's inequality ]\label{Popoviciu}
Let $(\Omega,\Sigma,\mathbb{P})$ be a probability space. Let $X\in\Omega\rightarrow[\alpha, \beta]$ be a random variable. Then, 
\be
\text{Var}(X)\leq \frac{(\beta-\alpha)^2}{4}.\\
\ee
\end{fact}
\end{appendices}
\end{document}